\RequirePackage{amsmath}
\RequirePackage{fix-cm}
\documentclass[twocolumn]{svjour3}          
\smartqed  
\usepackage{graphicx}
\def\makeheadbox{\relax}
\usepackage{amssymb}  
\usepackage{mathrsfs} 
\usepackage{tikz}
\usepackage{algpseudocode} 
\usepackage{algorithm}
\usepackage{url}
\newcommand{\acm}[3]{#1\rightarrow#2}
\newcommand{\ac}[3]{$#1\rightarrow#2$}
\newcommand{\omesi}{^\omega_\varsigma}

\usepackage{ifthen}

\newcommand{\currentScope}{}
\newcommand{\currentSort}{}
\newcommand{\currentSortLabel}{}
\newcommand{\currentAlign}{}
\newcommand{\currentSize}{}

\newcounter{la}
\newcommand{\TSetSortLabel}[2]{
  \expandafter\repcommand\expandafter{\csname TUserSort@#1\endcsname}{#2}
}
\newcommand{\TSort}[4]{
  \renewcommand{\currentScope}{#1}
  \renewcommand{\currentSort}{#2}
  \renewcommand{\currentSize}{#3}
  \renewcommand{\currentAlign}{#4}
  \ifcsname TUserSort@\currentSort\endcsname
    \renewcommand{\currentSortLabel}{\csname TUserSort@\currentSort\endcsname}
  \else
    \renewcommand{\currentSortLabel}{\currentSort}
  \fi
  \begin{scope}[shift={\currentScope}]
  \ifthenelse{\equal{\currentAlign}{l}}{
    \filldraw[process box] (-0.5,-0.5) rectangle (0.5,\currentSize-0.5);
    \node[sort] at (-0.2,\currentSize-0.4) {\currentSortLabel};
   }{\ifthenelse{\equal{\currentAlign}{r}}{
     \filldraw[process box] (-0.5,-0.5) rectangle (0.5,\currentSize-0.5);
     \node[sort] at (0.2,\currentSize-0.4) {\currentSortLabel};
   }{
    \filldraw[process box] (-0.5,-0.5) rectangle (\currentSize-0.5,0.5);
    \ifthenelse{\equal{\currentAlign}{t}}{
      \node[sort,anchor=east] at (-0.3,0.2) {\currentSortLabel};
    }{
      \node[sort] at (-0.6,-0.2) {\currentSortLabel};
    }
   }}
  \setcounter{la}{\currentSize}
  \addtocounter{la}{-1}
  \foreach \i in {0,...,\value{la}} {
    \TProc{\i}
  }
  \end{scope}
}

\newcommand{\TTickProc}[2]{ 
  \ifthenelse{\equal{\currentAlign}{l}}{
    \draw[tick] (-0.6,#1) -- (-0.4,#1);
    \node[tick label, anchor=east] at (-0.55,#1) {#2};
   }{\ifthenelse{\equal{\currentAlign}{r}}{
    \draw[tick] (0.6,#1) -- (0.4,#1);
    \node[tick label, anchor=west] at (0.55,#1) {#2};
   }{
    \ifthenelse{\equal{\currentAlign}{t}}{
      \draw[tick] (#1,0.6) -- (#1,0.4);
      \node[tick label, anchor=south] at (#1,0.55) {#2};
    }{
      \draw[tick] (#1,-0.6) -- (#1,-0.4);
      \node[tick label, anchor=north] at (#1,-0.55) {#2};
    }
   }}
}
\newcommand{\TSetTick}[3]{
  \expandafter\repcommand\expandafter{\csname TUserTick@#1_#2\endcsname}{#3}
}

\newcommand{\myProc}[3]{
  \ifcsname TUserTick@\currentSort_#1\endcsname
    \TTickProc{#1}{\csname TUserTick@\currentSort_#1\endcsname}
  \else
    \TTickProc{#1}{#1}
  \fi
  \ifthenelse{\equal{\currentAlign}{l}\or\equal{\currentAlign}{r}}{
    \node[#2] (\currentSort_#1) at (0,#1) {#3};
  }{
    \node[#2] (\currentSort_#1) at (#1,0) {#3};
  }
}
\newcommand{\TSetProcStyle}[2]{
  \expandafter\repcommand\expandafter{\csname TUserProcStyle@#1\endcsname}{#2}
}
\newcommand{\TProc}[1]{
  \ifcsname TUserProcStyle@\currentSort_#1\endcsname
    \myProc{#1}{\csname TUserProcStyle@\currentSort_#1\endcsname}{}
  \else
    \myProc{#1}{process}{}
  \fi
}

\newcommand{\repcommand}[2]{
  \providecommand{#1}{#2}
  \renewcommand{#1}{#2}
}
\newcommand{\THit}[5]{
  \path[hit] (#1) edge[#2] (#3#4);
  \expandafter\repcommand\expandafter{\csname TBounce@#3@#5\endcsname}{#4}
}
\newcommand{\TBounce}[4]{
  (#1\csname TBounce@#1@#3\endcsname) edge[#2] (#3#4)
}

\newcommand{\TState}[1]{
  \foreach \proc in {#1} {
        \node[current process] (\proc) at (\proc.center) {};
  };
}
\newcommand{\TCoopHit}[6]{
  \node[#2, apdot] at (#3) {};
  \foreach \proc in {#1} {
    \draw[#2,-] (#3) edge (\proc);
  }
  \path[hit] (#3) edge[#2] (#4#5);
  \expandafter\repcommand\expandafter{\csname TBounce@#4@#6\endcsname}{#5}
}

\usepackage{galois}

\tikzstyle{aS}=[every edge/.style={draw,->,>=stealth}]
\tikzstyle{Asol}=[draw,circle,minimum size=5pt,inner sep=0,node distance=1cm]
\tikzstyle{Aproc}=[draw,node distance=1cm]
\tikzstyle{Aobj}=[node distance=1.5cm]
\tikzstyle{Anos}=[font=\Large]
\tikzstyle{Assol}=[node distance=1.2cm]
\tikzstyle{AsolPrio}=[Asol,double]
\tikzstyle{AprocPrio}=[Aproc,double]
\tikzstyle{aSPrio}=[aS,double]

\newcommand{\startl}[1]{\node[Aproc] (#1) {$#1$};\node[Asol,right of=#1] (#1s) {};\path (#1) edge (#1s);}
\newcommand{\link}[2]{\node[Aproc,right of=#1s] (#2) {$#2$};\node[Asol,right of=#2] (#2s) {};\path (#1s) edge (#2) (#2) edge (#2s);} 
\newcommand{\specl}[3]{\node[Aproc,#1 right of=#2s] (#3) {$#3$};\node[Asol,right of=#3] (#3s) {};\path (#2s) edge (#3) (#3) edge (#3s);} 
\newcommand{\edl}[2]{\node[Assol,right of=#1s] (#1st){$\varnothing$};\path (#1s) edge (#1st);}

\usepackage{tikz}

\newdimen\pgfex
\newdimen\pgfem
\usetikzlibrary{arrows,shapes,shadows,scopes}
\usetikzlibrary{positioning}
\usetikzlibrary{matrix}
\usetikzlibrary{decorations.text}
\usetikzlibrary{decorations.pathmorphing}
\usetikzlibrary{arrows,shapes}

\definecolor{lightgray}{rgb}{0.8,0.8,0.8}
\definecolor{lightgrey}{rgb}{0.8,0.8,0.8}

\definecolor{lightred}{rgb}{1,0.8,0.8}
\definecolor{lightgreen}{rgb}{0.7,1,0.7}
\definecolor{darkgreen}{rgb}{0,0.5,0}
\definecolor{darkblue}{rgb}{0,0,0.5}
\definecolor{darkyellow}{rgb}{0.5,0.5,0}
\definecolor{lightyellow}{rgb}{1,1,0.6}
\definecolor{darkcyan}{rgb}{0,0.6,0.6}
\definecolor{lightcyan}{rgb}{0.6,1,1}
\definecolor{darkorange}{rgb}{0.8,0.2,0}
\definecolor{notsodarkred}{rgb}{0.8,0,0}

\definecolor{notsodarkgreen}{rgb}{0,0.7,0}

\colorlet{coloract}{darkgreen}
\colorlet{colorinh}{red}
\colorlet{coloractgray}{lightgreen}
\colorlet{colorinhgray}{lightred}
\colorlet{colorinf}{darkgray}
\colorlet{coloractgray}{lightgreen}
\colorlet{colorinhgray}{lightred}

\colorlet{colorgray}{lightgray}
\colorlet{colorhl}{blue}

\tikzstyle{boxed ph}=[]
\tikzstyle{sort}=[fill=lightgray, rounded corners, draw=black]
\tikzstyle{process}=[circle,draw,minimum size=15pt,fill=white,font=\footnotesize,inner sep=1pt]
\tikzstyle{gray process}=[process, draw=black, fill=lightgray]
\tikzstyle{highlighted process}=[current process, fill=gray]
\tikzstyle{process box}=[fill=none,draw=black,rounded corners]
\tikzstyle{current process}=[process, draw=black, fill=lightgray]
\tikzstyle{hl process}=[process,fill=blue!30]
\tikzstyle{tick label}=[font=\footnotesize]
\tikzstyle{tick}=[densely dotted] 
\tikzstyle{hit}=[->,>=angle 45]
\tikzstyle{selfhit}=[min distance=50pt,curve to]
\tikzstyle{bounce}=[densely dotted,>=stealth',->]
\tikzstyle{ulhit}=[draw=lightgray,fill=lightgray]
\tikzstyle{pulhit}=[fill=lightgray]
\tikzstyle{bulhit}=[draw=lightgray]
\tikzstyle{hl}=[very thick,colorhl]
\tikzstyle{hlb}=[very thick]
\tikzstyle{hlhit}=[hl]

\tikzstyle{update}=[draw,->,dashed,shorten >=.7cm,shorten <=.7cm]

\tikzstyle{unprio}=[draw,thin]
\tikzstyle{prio}=[draw,-stealth,double]

\tikzstyle{hitless graph}=[every edge/.style={draw=red,-}]

\tikzstyle{aS}=[every edge/.style={draw,->,>=stealth}]
\tikzstyle{Asol}=[draw,circle,minimum size=5pt,inner sep=0,node distance=1cm]
\tikzstyle{Aproc}=[draw,node distance=1.2cm]
\tikzstyle{Aobj}=[node distance=1.5cm]
\tikzstyle{Anos}=[font=\Large]

\tikzstyle{AsolPrio}=[Asol,double]
\tikzstyle{AprocPrio}=[Aproc,double]
\tikzstyle{aSPrio}=[aS,double]

\colorlet{colorhlwarn}{notsodarkred}
\colorlet{colorhlwarnbg}{lightred}
\tikzstyle{Ahl}=[very thick,fill=colorhlwarnbg,draw=colorhlwarn,text=colorhlwarn]
\tikzstyle{Ahledge}=[very thick,double=colorhlwarnbg,draw=colorhlwarn,color=colorhlwarn]

\tikzstyle{grn}=[every node/.style={circle,draw=black,outer sep=2pt,minimum
                size=15pt,text=black}, node distance=1.5cm, ->]
\tikzstyle{inh}=[>=|,-|,draw=colorinh,thick, text=black,label]
\tikzstyle{act}=[->,>=triangle 60,draw=coloract,thick,color=coloract]
\tikzstyle{inhgray}=[>=|,-|,draw=colorinhgray,thick, text=black,label]
\tikzstyle{actgray}=[->,>=triangle 60,draw=coloractgray,thick,color=coloractgray]
\tikzstyle{inf}=[->,draw=colorinf,thick,color=colorinf]
\tikzstyle{elabel}=[fill=none,text=black, above=-2pt,
minimum size=10pt, outer sep=0, font=\scriptsize, draw=none]

\tikzstyle{plot}=[every path/.style={-}]
\tikzstyle{axe}=[black,->,>=stealth']
\tikzstyle{ticks}=[font=\scriptsize,every node/.style={black}]
\tikzstyle{mean}=[thick]
\tikzstyle{interval}=[line width=5pt,red,draw opacity=0.7]

\tikzstyle{objective}=[process,very thick,fill=yellow!50]

\tikzstyle{coopupdate}=[-stealth,decorate,decoration={zigzag,amplitude=1.5pt,post=lineto,post length=.3cm,pre=lineto,pre length=.3cm}]

\tikzstyle{labelprio}=[circle, fill=blue!30, inner sep=0pt, minimum size=13pt]
\tikzstyle{labelprio1}=[labelprio]
\tikzstyle{labelprio2}=[labelprio, fill=red!60]
\tikzstyle{labelprio3}=[labelprio, fill=orange!50]
\tikzstyle{labelprio4}=[labelprio, fill=brown!50]

\tikzstyle{labelstocha}=[rectangle, rounded corners=4pt]

\tikzstyle{andot}=[circle, fill=black, inner sep=1.2pt, draw=transparent]
\tikzstyle{anligne}=[thick]

\tikzstyle{apdot}=[andot] 
\tikzstyle{apdotsimple}=[] 

\tikzstyle{equiv-externe}=[thick, rounded corners, draw=gray, fill=gray!10, align=center,
  inner sep=8]

 \tikzstyle{labeldelai1}=[circle, fill=red!60, inner sep=0pt, minimum size=8pt]
  \tikzstyle{labeldelai2}=[circle, fill=blue!30, inner sep=0pt, minimum size=8pt]
  \tikzstyle{labeldelai3}=[circle, fill=brown!50, inner sep=0pt, minimum size=8pt]
  \tikzstyle{labeldelai4}=[circle, fill=green!50, inner sep=0pt, minimum size=8pt]
  
\tikzstyle{local transitions}=[->,>=latex',thick,bend left=30,
               every node/.style={fill=white,inner sep=1pt,outer sep=1pt}]
\tikzstyle{reach}=[fill=lightgray,ellipse]

\tikzstyle{local transitions 2}=[->,>=latex',thick,bend left=100,
               every node/.style={fill=white,inner sep=1pt,outer sep=4pt}]

\tikzstyle{local transitions 3}=[->,>=latex',thick,bend right=100,
               every node/.style={ right, fill=white,inner sep=1pt,outer sep=4pt}]
               
\tikzstyle{vide}= [rectangle, minimum width=2em,minimum height=1.5em,]
\tikzstyle{stable}= [rectangle,fill=lightred]
\tikzstyle{current}= [rectangle,fill=lightcyan]
\tikzstyle{initial}= [rectangle,fill=green!20]

\tikzstyle{etiquette}=[midway,fill=blue!20,circle,scale=0.7pt]
\tikzstyle{etiquette2}=[midway,fill=green!20,circle,scale=0.7pt]
\tikzstyle{currentTrans}=[->,very thick,blue]
\tikzstyle{seperatedTransPart1}=[draw, thick, blue]
\tikzstyle{seperatedTransPart2}=[->, thick, blue]

\tikzstyle{mytext}=[thick, text width=4.5em,inner sep=1pt]
\tikzstyle{line} =[draw, thick, -latex',shorten >=2pt]
\tikzstyle{block} =[rectangle,text width=6em,draw,minimum height=4em, outer sep=0pt]

\tikzstyle{adn}=[every node/.style={circle,draw=black,outer sep=2pt,minimum
                size=15pt,text=black}, node distance=1.5cm, ->]
                
\tikzset{
    xmin/.store in=\xmin, xmin/.default=-3, xmin=-3,
    xmax/.store in=\xmax, xmax/.default=3, xmax=3,
    ymin/.store in=\ymin, ymin/.default=-3, ymin=-3,
    ymax/.store in=\ymax, ymax/.default=3, ymax=3,
}

\begin{document}

\title{A Heuristic for Reachability Problem in Asynchronous Binary Automata Networks
}
%



\author{Xinwei Chai\and Morgan Magnin \and Olivier Roux
}


\institute{X. Chai, M. Magnin and O. Roux\at
              Laboratoire des Sciences du Num\'erique de Nantes, UMR CNRS 6004 \'Ecole Centrale de Nantes, 1 rue de la No\"e - B.P. 92101 - 44321 Nantes Cedex 3, France \\
              Tel.: +33-782996065\\
              \email{\{xinwei.chai, morgan.magnin, olivier.roux\}@ls2n.fr}           
}

\date{April 2018}

\maketitle

\begin{abstract}
On demand of efficient reachability analysis due to the inevitable complexity of large-scale biological models, this paper is dedicated to a novel approach: PermReach, for reachability problem of our new framework, Asynchronous Binary Automata Networks (ABAN). ABAN is an expressive modeling framework which contains all the dynamics behaviors performed by Asynchronous Boolean Networks. Compared to Boolean Networks (BN), ABAN has a finer description of state transitions (from a local state to another, instead of symmetric Boolean functions). To analyze the reachability properties on large-scale models (like the ones from systems biology), previous works exhibited an efficient abstraction technique called Local Causality Graph (LCG). However, this technique may be not conclusive. Our contribution here is to extend these results by tackling those complex intractable cases \textit{via} a heuristic technique. To validate our method, tests were conducted in large biological networks, showing that our method is more conclusive than existing ones.
\end{abstract}

\section{Introduction}
\label{intro}
Works on concurrent systems have been of interest for systems biology for a decade \cite{bockmayr2002using,bortolussi2008modeling,wiley2003computational}. In this context, the challenges nowadays consist of not only model validation with regard to existing knowledge on systems but also behavior prediction of these systems. 
With quantities of available data provided by new technologies, \textit{e.g.} DNA microarray \cite{marx2013}, there is a growing need for high-performance analytic tools, especially for reachability problem, as many static and dynamical properties are transformable to the reachability of certain states. Reachability problem has been studied under many different modeling frameworks for decades \cite{akutsu2007control,barrett2006complexity,Daws1998,esparza1998,mayr1984,wozna2003} and takes an important part in Model Checking \cite{clarke20142}. State Space Explosion problem arises in reachability analysis of concurrent systems as the state space is exponential to the number of components in the model, thus disables naive approaches. More concretely, Peterson has shown that the reachability problem of Petri net is exponential time-hard and exponential space-hard \cite{peterson1977petri}, and this conclusion does not change even in some special situations \cite{esparza1998}. We are prone to believe that the reachability of Boolean Network is also in this class, as there does not exist an efficient solution of polynomial complexity, although there does not exist such formal proof for the non-existence.

Related studies have been carried over various frameworks: Plateau et al. \cite{plateau1991stochastic} propose a Stochastic Automata Network and study its steady-state behavior, while the reachability analysis is absent; Li et al. \cite{li2012reachability,li2014stability} investigate theoretically the stability, controllability and reachability of Switched Boolean Networks, but their method remains computationally expensive; Ben Abdallah et al. \cite{abdallah2015exhaustive} designed an exhaustive algorithm for reachability using ASP (Answer Set Programming) \cite{baral2003knowledge}. Although ASP has a built-in optimization, the complexity is still exponential. 

To tackle the persisting complexity issue, symbolic model checking \cite{burch1992symbolic} and SAT-solvers \cite{abdulla2000symbolic} have been considered over years, but the solution space of original problem remains huge. Paulev\'e \textit{et al.} \cite{folschette2015,pauleve2011} have proposed new discrete modeling frameworks for concurrent systems: Process Hitting and its updated Automata Network (AN) form. They provide an approach to address this issue by designing a static abstraction (with an over-approximation and an under-approximation of the real dynamics) inspired by abstract interpretation: Local Causality Graph (LCG). This static analysis drastically reduces the state-space and avoids costly global search \cite{pauleve2012}.

In various circumstances, LCG is capable of giving a deterministic result of the reachability of desired states and corresponding realizing transition sequences (if reachable) in polynomial time to the number of automata \cite{pauleve2016goal}, but its applicability is still limited. There are inconclusive cases which disable the reasoning of sufficient conditions, if there exists cycles or conflicts in the LCG. We will identify these cases later (in section \ref{limitation}).
\subsection*{Main Contributions}
With the initiative of LCG, this paper is devoted to the study of general reachability problems in Asynchronous Binary Automata Networks (ABAN), then to gain a more profound understanding of the dynamics of biological systems. Many biological networks are encoded in Boolean style \cite{kauffman1969}, because BN is a simple formalism but with strong expressiveness also due to the imprecision of raw data. Different from BN, ABAN is a finite state machine comprised by communicating automata. Each automaton has 2 states corresponding to the bit 0 and 1 in BN. One of the interests of studying ABANs is that BN may be not expressive enough in biological context. For example, to model the dynamic behavior ``$a\gets$ $1$ when $b=1$'', we have $a(t+1)=b(t)$ in BN, $a$ always follows the evolution of $b$, but with an unwanted behavior ``$a\gets 0$ when $b=0$''. ABAN can model this dynamics as \textit{via} \ac{b_1}{a_0}{a_1} without redundancy. More importantly, ABAN reduces the complexity of reachability analysis under some special conditions (\textit{e.g.} LCG with cycles) and helps to embark on the most difficult inconclusive instances without globally traversing the state space, which is the main theme of this paper. Besides, BNs are transformable to Automata Networks, and this property makes our work more extensive (see Appendix \ref{appendix:C}).

To solve a problem which is costly for naive approaches, there are basically three methodologies: 
\begin{enumerate}
\renewcommand{\labelenumi}{(\theenumi)}
\item abstract the original problem or find a simplified formalism
\item study polynomially solvable cases of the problem
\item apply heuristic
\end{enumerate}
(1) guarantees a correct solution for the problem after abstraction which is not necessarily equivalent to the solution of the original problem; (2) guarantees a correct solution only for a part of instances; (3) keeps the original problem but does not guarantee a correct solution. We are going to walk on the pathway (3) of heuristic. After diving into the mechanics of LCG and the inconclusive cases, we figure out the reason why those cases are intractable by existing static reasoning. With a better understanding of the internal structure of LCGs, we develop a heuristic technique aiming at the application for general instances. This heuristic method has a better performance on conclusiveness than static reasoning, because it attempts to explore a part of the system dynamics \textit{via} partial verification. In the end, we conduct tests on signaling networks of around 100 components (TCR and EGFR, see Section \ref{sect:5}): the results of LCG contain inconclusive instances \cite{folschette2015} while our new method solves them.

This paper is organized as follows: in section 2, we will introduce the formal background, BN and Asynchronous Binary Automata Network (ABAN); section 3 presents the analysis of dynamics using only static reasoning; 
section 4 is the core content of this paper, concerning the solution of inconclusive cases; discussion about tests and conclusion are placed in section 5 and 6.

\section{Preliminaries}\label{sect:2}
\textit{Notations}:
$::$ sequential connector;
$\Rsh$ state change;
$\land \bigwedge$ logic AND;
$\lor\bigvee$ logic OR;
$\#$ cardinal;
$a.$next the successor of $a$. 

%

Asynchronous Binary Automata Network (ABAN) is a variant of traditional AN. Binary means that every automaton has exactly two possible states $(0,1)$ and asynchronous implies the update scheme with no more than one automaton can change its value at a time. 

\begin{definition}[ABAN]
An ABAN is a triplet $AB = (\mathbf{\Sigma},\mathbf{L},\mathbf{T})$, where:
\begin{itemize}
\item $\mathbf{\Sigma}\triangleq\{a,b,\ldots,z\}$ is the finite set of automata with every component having a Boolean state;
\item $L_a\triangleq\{a_0,a_1\}$ is the set of binary states of automaton $a\in \mathbf{\Sigma}$, $\mathbf{LS}=\underset{a\in \mathbf{\Sigma}}{\cup} L_a$ is the set of all local states, and $\mathbf{L}\triangleq \underset{a\in \mathbf{\Sigma}}{\times} L_a$ is the set of global states, the state of automaton a at state s is denoted $s[a]=a_i$;
\item $\mathbf{T}\triangleq \{A\rightarrow b_i\mid b\in \mathbf{\Sigma} \land A\in \mathbf{L}\}$ is the set of transitions, where $A$ is the required state(s) for the transition, which allows to flip $b_i$ to the other Boolean state. In other words, transition $tr=\acm{A}{b_j}{b_k}$ is said firable iff $A\subseteq s$.
\end{itemize}
\end{definition}

Furthermore, to describe the evolution of an ABAN, we use the notion of trajectory:
\begin{definition}[Trajectory]
Given initial state $s\in \mathbf{L}$, a trajectory $\delta$ from $s$ to $\Omega$ is a sequence of transitions in $\mathbf{T}$ that can be fired successively. 
\end{definition}

From a given initial state $s$, the state after firing $\delta$ is denoted $s\cdot \delta$ and its local form of certain automaton $a$ is noted $(s\cdot \delta)[a]$.
Fig \ref{fig:1} shows an example of ABAN, with initial state $s=\langle a_0,b_0,c_0,d_0,e_0\rangle$ and a possible trajectory is $\delta=\acm{d_0}{b_0}{b_1}::\acm{b_1}{d_0}{d_1}::\acm{d_1}{c_0}{c_1}::\acm{\{b_1,c_1\}}{a_0}{a_1}$. After firing $\delta$, final state $\Omega=s\cdot \delta=\langle a_1,b_1,c_1,d_1,e_0\rangle$.

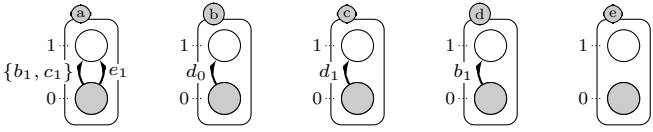
\begin{figure}[ht]
\centering
\begin{tikzpicture}[apdotsimple/.style={apdot},scale=0.7, every node/.style={scale=0.8}]
\scriptsize
\TSort{(0,0)}{a}{2}{l}
\TSort{(2.5,0)}{b}{2}{l}
\TSort{(5,0)}{c}{2}{l}
\TSort{(7.5,0)}{d}{2}{l}
\TSort{(10,0)}{e}{2}{l}

\path[local transitions]

	(a_0) edge node[auto] {\{$b_1, c_1$\}} (a_1)
    (a_0) edge[bend right] node[right] {$e_1$} (a_1)
	(b_0) edge node[auto] {$d_0$} (b_1)
	(d_0) edge node[auto] {$b_1$} (d_1)
	(c_0) edge node[auto] {$d_1$} (c_1)
;

\TState{a_0, b_0, c_0, d_0,e_0}

\end{tikzpicture}
\caption{Example of ABAN}\label{fig:1}
\end{figure}	
As to the reachability problem, given an ABAN, global reachability can be formalized as: global state $\Omega$ is reachable iff there exists a trajectory $\delta$ such that $s\cdot \delta=\Omega$. Partial reachability is defined analogously: local state $\omega=a_i$ is reachable iff there exists a trajectory $\delta$ such that $(s\cdot \delta)[a]=a_i$. Reachability is denoted $reach (\Omega)$ or $reach(\omega)$ and takes Boolean value $0$ or $1$. In Fig \ref{fig:1}, we can see $\Omega=\langle a_1,b_1,c_1,d_1,e_0\rangle$ or $\omega=a_1$ is reachable from initial state $s$ \textit{via} trajectory $\delta$, that is $reach(a_1)=1$ and $reach(\Omega)=1$. 

To simplify the notations, all the initial states of ABAN are set to 0 by default.
\section{Static analysis of reachability property}\label{sect:3}
To approach various dynamical properties of such networks, Local Causality Graph (LCG) is an efficient static analytic tool for reachability put forward by Paulev\'e \textit{et al.} \cite{pauleve2011}. LCG determines the existence of trajectory of the desired state without global verification.

LCG functions as follows: its over-approximation and under-approximation which give respectively a necessary condition and a sufficient condition of reachability. With these conditions, we can conclude in many cases.

\begin{table}[ht]
\centering
\begin{tabular}{c|c|c|c|c}
Over&True&True&False&False\\
\hline
Under&True&False&True&False\\
\hline
Reach&True&Inconclusive&Impossible&False\\
\end{tabular}
\caption{Truth table of LCG}
\label{tab:1}
\end{table}

More importantly, LCG is also able to provide us with a trajectory if $\omega_i$ is reachable suggested by under-approximation. We are not going to detail the original version of this trajectory finding technique, instead, we propose its adaptation for ABAN in order to approach the solution of inconclusive cases.
In this paper, only binary networks are studied, so we propose a simplified form of LCG instead of two LCGs (over and under-approximation) which is well suited for the present need.

The drawback is also clear: there are inconclusive cases, which means LCG is neither able to solve in this situation, nor able to generate a trajectory (if reachable). To improve the conclusiveness (less inconclusive cases) of this method, we generalize the over- and under-approximation into one (SLCG) and add a more detailed reachability analysis.

Besides, LCG is a technique designed for Automata Networks \cite{folschette2015}. To give it a wider applicability, in appendix \ref{appendix:C}, we can see that any BN is transformable to ANs and then SLCG is applicable to its analyses of dynamical properties.
\subsection{Simplified Local Causality Graph (SLCG)}

\begin{definition}[SLCG]\label{defLCG}
Given ABAN $AB = (\mathbf{\Sigma},\mathbf{L},\mathbf{T})$, initial state $\varsigma$ and a desired local state $\omega$, SLCG $A\omesi= (V\omesi,E\omesi)$ is the smallest recursive structure with $V\omesi \subseteq \mathbf{LS}\cup \mathbf{Sol}$ and $E\omesi \subseteq V\omesi\times V\omesi$ which satisfies:
\begin{eqnarray*}
\omega&\subseteq& V\omesi \\
a_i\in V\omesi\cap \mathbf{LS} &\Leftrightarrow& \{ (a_i, sol_{a_i})| a_i\in \varsigma\}\subseteq E\omesi \\
sol_{a_i}\in V\omesi \cap \mathbf{Sol}&\Leftrightarrow& \{ (sol_{a_i},\mathbf{V}_a (sol_{a_i})\}\subseteq E\omesi
\end{eqnarray*}
Notations: $\mathbf{Sol}$ is the set of solutions and $\mathbf{V}$ is the set of required local states of $sol_{a_i}$.
\end{definition}
Intuitively, when the recursive construction is complete, SLCG is in fact a digraph with state nodes $\mathbf{LS}$ and solution nodes $\mathbf{Sol}$. $E$ consists of the links between state nodes and solution nodes. To access certain local states, at least one of its successive solution (corresponding transitions form solution nodes) needs to be fired; similarly, to make one solution node firable, all of its local states need to be satisfied. A recursive reasoning of reachability begins with a state node representing desired local state, go through $a_i\mapsto sol_{a_i}\mapsto b_j \cdots$ and end with initial state (reachable) or a local state without solution successor (unreachable).

\begin{figure}[!ht]
\centering
\begin{tikzpicture}[aS,scale=0.9, every node/.style={scale=0.9}]  
  	
  	\startl{a_1};
    \node[Asol,left of=a_1] (a_1s1){};
    \node[Aproc,left of=a_1s1] (e_1){$e_1$};
    \path 
    (a_1s1) edge (e_1)
    (a_1) edge (a_1s1)
    ; 
  	\specl{above}{a_1}{b_1};
  	\link{b_1}{d_0};
  	\edl{d_0};
  	\specl{below}{a_1}{c_1};
	\link{c_1}{d_1};
    \path (d_1s) edge (b_1);
    \end{tikzpicture}
\caption{SLCG for calculating the reachability of $a_1$ of the ABAN in Fig \ref{fig:1}, with the squares representing local states and small circles representing solution nodes}
\label{fig:2}
\end{figure}
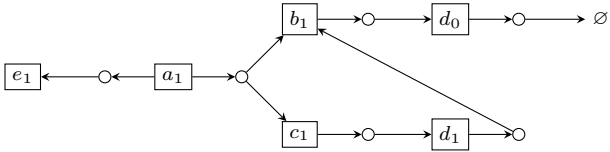
In Fig \ref{fig:2}, the reachability of $a_1$ is computed locally. The left solution node of $a_1$ does not lead to the goal because its successor $e_1$ does not have any successive solution node, \textit{i.e.} $e_1$ is unreachable; the right solution node of $a_1$ requires $b_1$ and $c_1$, and they finally lead to $d_0$ then to $\varnothing$, which is to say, we need nothing to reach $d_0$ as $d_0$ is in initial state (trivial solution).

We can figure out state nodes act as an OR gates while solution nodes act as AND gates. The reachability suggested by SLCG $reach'$ (different from real reachability $reach$) is computed recursively as follows:
\begin{align*}
reach'(a_1)&=reach'(e_1)\lor (reach'(b_1)\land reach'(c_1))\\
&=reach'(b_1)\land reach'(d_1)=reach'(b_1)\\
&=reach'(d_0)=1
\end{align*}

The algorithm of SLCG construction is in Appendix \ref{appendix:B}.


\subsection{Limitation of SLCG}\label{limitation}
Although SLCG allows us to reason the reachability locally without traversing the whole state space, it is still providing us with a necessary condition (quasi-equivalent) of reachability because SLCG does not simulate the real evolution of the system.

The inequivalent condition does not suggest it is impossible to reveal the real dynamics of the system. We are going to show that an SLCG gives an equivalent condition of reachability iff it satisfies the following conditions:

\begin{enumerate}
\item No cycles in SLCG
\item No conflicts in SLCG

\end{enumerate}
To be more formal, a cycle (1) is in the form of $a_i\mapsto\cdots\mapsto a_i$, i.e. to access $a_i$, we have to reach first $a_i$. This self-involvement makes the reachability inconclusive. A conflict (2) is that a solution node has multiple successors generating branches, and there are different states of the same automaton \textit{i.e.} $a_i$ and $a_{\lnot i}$. We can not decide the order of reaching these states, because reaching one state may disable the reachability of another one. Sometimes there exists a trajectory which accesses these states in certain order, sometimes there does not exist such.

In the following examples, if we ignore those restrictions, SLCG does not imply real reachability, nor it is possible to extract a trajectory.

\begin{enumerate}
\item Example of 
Fig \ref{fig:2}, although there is a conflict, $a_1$ is reachable.
\item $\mathbf{\Sigma}=\{a,b,c\}$, $\mathbf{T}=\{\acm{b_0}{a_0}{a_1},\ \acm{a_0}{b_0}{b_1},\ \acm{\{a_1,b_1\}}{c_0}{c_1}\}$,  desired final state $\omega=a_1$. Both $a_1$ and $b_1$ are reachable, but they can not be reached simultaneously. In the SLCG, there are two branches, $a_1\mapsto b_0$ and $b_1\mapsto a_0$, the automata $a$ and $b$ involve themselves in different branches, where a conflict appears.
\begin{figure}[ht]
\centering
\begin{tikzpicture}[aS,scale=0.9, every node/.style={scale=0.9}]  
  	
  	\startl{c_1};
  	\specl{above}{c_1}{a_1};
  	\link{a_1}{b_0};
  	\edl{b_0};
  	\specl{below}{c_1}{b_1};
	\link{b_1}{a_0};
  	\edl{a_0};
    \end{tikzpicture}
\caption{The SLCG of example 2}
\label{fig:3}
\end{figure}
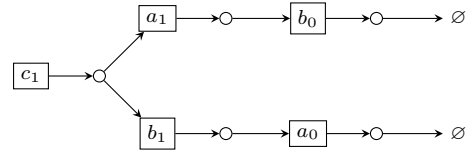
\item $\mathbf{\Sigma}=\{a,b,c\}$, $\mathbf{T}=\{\acm{b_0}{a_0}{a_1},\acm{a_0}{b_0}{b_1},\acm{\varnothing}{a_1}{a_0}, \acm{\varnothing}{b_1}{b_0},\acm{\{a_1,b_1\}}{c_0}{c_1}\}$, $\omega=a_1$. Similarly to example 2, both $a_0$ and $b_0$ are reachable, but they can not be reached simultaneously. In this example, we can see a cycle in the SLCG $a_1\mapsto b_0\mapsto b_1 \mapsto a_0\mapsto a_1$.
\begin{figure}[ht]
\centering
\begin{tikzpicture}[aS,scale=0.9, every node/.style={scale=0.9}]  
  	
  	\startl{c_1};
  	\specl{above}{c_1}{a_1};
  	\link{a_1}{b_0};
  	\edl{b_0};
  	\specl{below}{c_1}{b_1};
	\link{b_1}{a_0};
  	\edl{a_0};
    \path
    	(a_0s) edge (a_1)
        (b_0s) edge (b_1)
    ;
    \end{tikzpicture}
\caption{The SLCG of example 3}
\label{fig:4}
\end{figure}
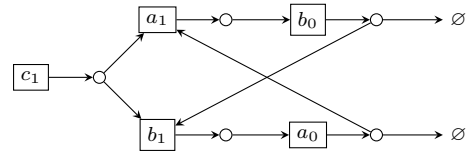

\end{enumerate}

In Example 1, $a_1$ is reachable, while in Example 2 and Example 3, $a_1$ is unreachable. This inconclusiveness is a limitation of SLCG.

As there are cycles in SLCG generated by feedback loops in biological regulatory networks, the existing approach does not allow a solution generally.

Even though we broadened the applicability of SLCG, this method is still not universally applicable for all ABANs or BNs due to the limitations. From the former examples, we realize that it is difficult to solve the reachability problem in general. In the rest of this paper, we are going to discuss how to improve the performance of existing approaches and our new methods. 

\subsection{Trajectory extraction for SLCG}\label{sectExtract}
In this section, we will first prove that if an AN is binary, reachability problem becomes equivalent to $reach'$ suggested by SLCG, with the restrictions in section \ref{limitation} satisfied. With the equivalence, we propose a method to find a trajectory starting from the initial state and ending with desired final state.

\begin{theorem}
In an ABAN, partial reachability and the reachability obtained by SLCG are equivalent, iff there is no cycle nor conflict.
\end{theorem}
\begin{proof}
Necessity: SLCG gives a necessary condition of reachability, thus all reachable local states in ABANs satisfy their reachability suggested by SLCG.

Sufficiency: in SLCG, the reachability is computed recursively.
\begin{enumerate}
\item $\forall a_i\in \mathbf{LS},\ reach(a_i)=\bigvee reach(sol_{a_i})$, as the reachability of any solution
of $a_i$ makes $a_i$ reachable.\label{1}
\item As there is no conflict in the SLCG, $\forall sol\in \mathbf{Sol}$, for all relating transitions $tr=\acm{b_i}{c_j}{c_k}$ of $sol$, $b$ are different automata from each other. Firing any $tr$ will not influence the firability of other transitions. Therefore $reach(sol)=\bigwedge reach(b_i)$.
\end{enumerate}
Furthermore, on one hand, as there is no cycle in the SLCG, in every iteration $reach(a_i)=\bigvee reach(sol_{a_i})$ and $reach(sol)=\bigwedge reach(b_i)$, $b_i$ are newly appeared local states, otherwise claiming $b_i$ will form a cycle $b_i\mapsto sol_{b_i}\mapsto \cdots b_i$. On the other hand, $\mathbf{LS}$ is finite, the iteration will come to an end where there is no more solution node (unreachable) or initial state (reachable).	
\qed\end{proof} 

Note that for any local initial state $a_i$, $reach(a_i)=reach'(a_i)=1$; similarly, for any local states $b_j$ without successor solution node, $reach(b_i)=reach'(b_i)=0$. As we can see the reasoning for $reach$ and $reach'$ are identical, we have $reach=reach'$, \textit{i.e.} partial reachability and the reasoning of SLCG are equivalent.

Algorithm \ref{algorithm:1} shows exactly how the trajectory is formed in depth-first order (find one valid then stop searching) after computing the reachability of every node in SLCG:

\begin{algorithm}[ht]
\begin{algorithmic}
\State Initialization: trajectory $\delta\gets\varnothing$, $visited=\varnothing$
\State Input: desired state $a_i$
\Function{extract}{$a_i$, $\delta$}
	\State $sol\gets $random($a_i{\rm .next}$)
	\State $\delta \gets (\acm{sol{\rm .next}}{a_{\lnot i}}{a_i}::\delta$)
	\If{$a_i\in visited$}
    	\State \Return $\delta$
    \EndIf
	\If{$sol{\rm .next}\neq \varnothing$}
    	\For{$b_j\in sol{\rm .next}$}
        	\State$visited\gets visited\cup b_i$
            \State $\delta\gets \delta::$\Call{extract}{$b_j$, $\delta$}
        \EndFor
    \EndIf
\EndFunction
\end{algorithmic}
\caption{Trajectory-extraction}
\label{algorithm:1}
\end{algorithm}

Let us take the example in Fig \ref{fig:2}, even though there is a conflict ($d_0$ and $d_1$ appear in different branches of solution node of $a_1$). Start from $a_1$, the only reachable solution node requires $b_1\land c_1$, for now $\delta=\acm{\{b_1,c_1\}}{a_0}{a_1}$. Begin with the branch $c_1$, we have $\delta=\acm{d_0}{b_0}{b_1}::\acm{b_1}{d_0}{d_1}::\acm{d_1}{c_0}{c_1}::\acm{b_1,c_1\}}{a_0}{a_1}$, branch $c_1$ ends. Continue with branch $b_1$, we find out $b_1\in visited$ is already reached, branch ends. The trajectory is the same as that in Fig \ref{fig:1}.

\section{Heuristic for inconclusive cases}\label{sect:4}

SLCG loses its generality of solution due to the existence of AND gates. 
In this section, we are trying to make a little compromise in exactitude in order to have a general solution of reachability problem. 
As the approach contains permutations, we limit the in-degree $I$ of its corresponding BN to $O (1)$ and we limit the number of the clauses of each CNF (see appendix \ref{appendix:C}) to 2. As a consequence, the out-degree of corresponding AND gates in SLCG is bounded to $O (1)$ and that of OR gates is bounded to $2$. This hypothesis is reasonable, because in ordinary biological networks, every component interacts only with a small part of the whole network \cite{akutsu2007control}.

To face with an arbitrary ABAN in this context, for now, the remaining tasks are cycles and conflicts in its SLCG.
\subsection{Preprocessing of SLCG}\label{sectprecond}
\subsubsection{Detection of cycles}
To be more precise, the notion of cycle can be expanded to Strongly Connected Components (SCC) of size greater than 1 as an SCC may contain several nested cycles. In \cite{tarjan1972}, the search of SCCs can be done in $O (|V|+|E|)$ time. As SLCG is in fact a sparse graph (the out-degree is limited to $O (1)$), the search of SCCs can be done in $O (|V|)$, \textit{i.e.} linear time.
\subsubsection{Treatment of cycles}
LCG, the former version of SLCG is not capable of concluding due to the existence of cycles \cite{folschette2015}. However SLCG is able to do so for its Boolean properties, and it shows exactly the reachability in the semantics of ABAN. Besides, cycles are no longer the reason of inconclusiveness.

\begin{theorem}\label{cycletheo}
Given an SLCG of ABAN, if it possesses no cycle with \textit{fork} (state node that has plural successive solution nodes, \textit{i.e.} no OR gate in the cycle), then all the local states in its cycles are unreachable and can then be deleted from the SLCG.
\end{theorem}

\begin{proof}
Suppose an arbitrary cycle $C=a_i\mapsto \cdots b_j\mapsto\cdots \mapsto a_i$,with $\mapsto$ an arrow in SLCG. Note that $reach(b_j)=reach(b_j.\text{next})=reach(b_j.\text{next}.\text{next})\cdots$. As $C$ is a cycle, the reachability of all the local states are equivalent. Reaching any element $b_j$ in $C$ implies the reachability of all the elements in $C$. In SLCG, the reachability is deducted by reaching the initial state, \textit{i.e.} if certain elements in $C$ are reachable, there exists at least one element $b_j$ belonging to the initial state. However $b_j$ should have no successor because it is reached already. This fact reveals that the fork containing $b_j$ can never form a cycle, contradiction. So none of the element in $C$ is reachable.
\qed\end{proof}

With this theorem, we can also deal with cycles containing forks. 
\begin{lemma}\label{cyclelem}
The reachability of the elements in a cycle with forks equals to the disjunction of the reachability of the forks.
\end{lemma}

\begin{proof}
Suppose a cycle $C=a_i\mapsto \cdots b_j\mapsto \circ\mapsto c_k \mapsto\cdots \mapsto a_i$, where $\circ$ represents a solution node. Suppose there is one fork located at $b_j$. $b_j.\text{next}=\{c_k,\{d_l\}\}$, where $\{d_l\}$ are outsiders of the cycle. According to the reasoning of SLCG, $reach (b_j)=reach (c_k)\lor (\bigvee reach (d_{l}))$. As in the proofs above, all the local states in the cycle share the same reachability: $reach (b_j)=reach (c_k)\Rightarrow reach (b_j)=reach (b_j)\lor (\bigvee reach (d_{l}))$. To keep this equation always valid, there must be $reach (b_j)=\bigvee reach (d_{l})$. Similarly, we can obtain the reachability in a cycle with plural forks: $reach (C)=\bigvee reach (forks)$.
\qed\end{proof}

With Theorem \ref{cycletheo} and Lemma \ref{cyclelem}, before stepping into the next part of dealing with AND gates, we can perform a recursive preprocessing by deleting the cycles in SLCG to ensure no cycle remaining.
\subsection{AND Gates in SLCG}\label{sectAndGates}
After preprocessing, we can get rid of cycles. The guideline is then to analyze an SLCG with only AND gates. To achieve this goal, we need to find a trajectory reaching all the components of given AND gates simultaneously. These components form a sub-state, and if the sub-state is reachable, the corresponding transition of AND gate can be fired. 
\begin{definition}[sub-state]
The set of sub-states $S$ is the Cartesian product of the local states of several automata: $S\triangleq \underset{a\in \mathbf{\Sigma}'}{\times} L_a$, where $\mathbf{\Sigma}'\subseteq \mathbf{\Sigma}$.
\end{definition}

Example: in Fig \ref{fig:1}, $sub=\{ b_1,c_1\}$ is a sub-state, when $sub$ is reached, transition \ac{\{ b_1,c_1\}}{a_0}{a_1} is firable.

As the cycles do not persist, the order reaching the members in a sub-state is the only factor that affects the final reachability. The reachability of a sub-state can be then formulated as sequential reachability:
\begin{definition}[Sequential reachability]
Let $sub=\{ls_1,\ldots,ls_n\}$ and sequence $seq=ls_1::\ldots::ls_n$, the sequential reachability of $sub$ is denoted $reach (seq)=reach (ls_1)::\ldots::reach (ls_n)$, i.e. from initial state, the $sub$ is reachable in the order $seq$ by following the trajectories given by SLCG.
\end{definition}

Example: Fig \ref{fig:5} shows the SLCG for reachability of $c_1$ in ABAN with transitions $\mathbf{T}=\{\acm{\{a_1,b_1\}}{c_0}{c_1},\acm{b_0}{a_0}{a_1},\acm{c_0}{b_0}{b_1}\}$.
\begin{figure}[ht]
\centering
\begin{tikzpicture}[aS]  
  	
  	\startl{c_1};
  	\specl{above}{c_1}{a_1};
  	\link{a_1}{b_0};
  	\edl{b_0};
  	\specl{below}{c_1}{b_1};
	\link{b_1}{c_0};
  	\edl{c_0};
    \end{tikzpicture}
\caption{Reachability depends on firing order}
\label{fig:5}
\end{figure}
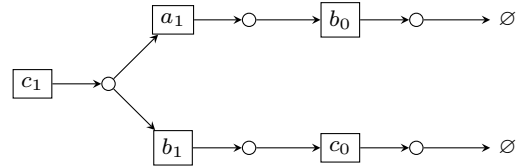

$a_1$ and $b_1$ are reachable respectively but is not necessarily for $c_1$. If we begin with the fork $a_1$, $sub=\{a_1,b_1\}$ is reachable with trajectory $\acm{b_0}{a_0}{a_1}::\acm{c_0}{b_0}{b_1}::\acm{\{a_1,b_1\}}{c_0}{c_1}$. However if we begin with the fork $b_1$, after firing $\acm{c_0}{b_0}{b_1}$, $b_0$ is no longer reachable, resulting the unreachability of $a_1$. We have $reach (a_1::b_1::c_1)=1$ and $reach (b_1::a_1::c_1)=0$.

As the firing order matters, we come to verify all the possible sequential reachabilities of certain sub-state to obtain its reachability.

\begin{theorem}\label{theoperm}
Given sub-state $sub=\{ls_1,\ldots,ls_n\}$, with all the local states in $sub$ are reachable: $reach (ls_i)=1,\ \forall i\in[1,n]$, the set of permutations of $sub$ is denoted $Perm (sub)=\{(ls_1::ls_2,::\ldots ::ls_n),\ \cdots,\ (ls_n::ls_{n-1}::\ldots,::ls_1)\}$. $\bigvee reach (Perm (sub))=1$ is a quasi-equivalent condition of $reach (sub)=1$.
\end{theorem}
\begin{proof}
Notation: $a_i\triangleright b_j$ means that $a_i$ must be present before $b_j$.

Necessity: if there exists a permutation $perm_i\in Perm$ s.t. $reach (perm_i)=1$, then $sub$ can be reached according to $perm_i$.

Quasi-sufficiency: in Definition \ref{defLCG}, SLCG is the smallest structure which leads to desired local state. To reach $sub$, every local state in SLCG is mandatory to be reached. Then the question of sufficiency becomes:

Given $\bigvee reach (Perm (sub))=0$, does there exist a permutation $perm (Ls)$ s.t. $reach (sub)=1$?

Although all the local states in $sub$ are reachable, the existence of conflicts leads to potential unreachability of $sub$. Some conflicts are solvable, see Fig \ref{fig:5}. There are also unsolvable ones. An unsolvable conflict can be formalized as: $ls_1\triangleright ls_2 \triangleright \cdots \triangleright ls_1$, where $ls_1$ is before $ls_2$ and $ls_2$ is before $ls_1$, which is impossible in asynchronous semantics. Example: in Fig \ref{fig:3}, to reach $\{a_1,b_1\}$, $a_1\triangleright b_1\triangleright a_1$, so the sub-state is unreachable. We can see the conflict is unsolvable no matter how we change the order of firing. For solvable conflicts, $perm(sub)$ probably covers one of the admissible order. One possible counterexample is shown in Fig. \ref{FigConflictInForks}. 
\qed\end{proof}

It is remarkable that the former approach is efficient in deciding reachability and finding reaching trajectory, but it has a drawback: if there exists a solvable conflict in different forks, traversing permutations may be not able to find the trajectory towards goal state. In Fig \ref{FigConflictInForks}, if $sol_{c_1}$ is resolved first, automaton $d$ will be on state $d_1$, which disables the reachability of $b_1$. In other cases, the trajectory of $a_1$ is findable.
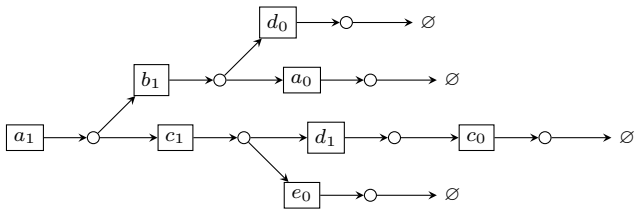
\begin{figure}[ht]
\centering
\begin{tikzpicture}[aS,scale=0.9, every node/.style={scale=0.9}]  
  	
	\startl{a_1};
	\specl{above}{a_1}{b_1};
	\specl{above}{b_1}{d_0};
	\edl{d_0};
	\link{b_1}{a_0};
 	\edl{a_0};
	\link{a_1}{c_1};
 	\link{c_1}{d_1};
 	\link{d_1}{c_0};
 	\edl{c_0};
 	\specl{below}{c_1}{e_0};
 	\edl{e_0};
    \end{tikzpicture}
\caption{Conflicts in different forks}\label{FigConflictInForks}
\end{figure}

For one sub-state, the number of permutations is $\mathbf{A}_I^I=I!$, as $I=O (1)$, this number is adjustable by controlling $I$ ($7!=5040,10!=3.6\times 10^6$).

However there is more than one AND gate in general biological networks and those AND gates could be chained \textit{e.g.} \cite{samaga2009logic}, \textit{i.e.} the successors of certain AND gate contain other AND gates. We analyze first the simple AND gates $simp$ (without successive AND gates) by traversing its permutations. If all elements in $simp$ are reachable, update initial state by firing all the transitions in extracted trajectory \textit{via} (section \ref{sectExtract}), and delete the successors of $simp$, then restart the process from finding simple AND gates. During the whole process, if an AND gate is not reachable after traversing its permutations, the final goal state is not reachable as the SLCG is linked by logical AND. Otherwise, when the process terminates, there is no AND gate, SLCG is conclusive as there is no cycle or conflict.

The statement above is the worst case: in reality, all AND gates are not necessarily composed of exact $I$ components, and permutations are determined to be unreachable before verification as its subsets may have been confirmed unreachable in other tentatives.

For example: given an AND gate $sol_a=b\land c\land d$, where $b,c,d$ are local states. Normally 6 realizing orders need checking: $b::c::d$, $b::d::c$, $c::b::d$, $c::d::b$, $d::b::c$ and $d::c::b$. If we find the order $b::c$ is not realizable when verifying the first realizing order, then we do not have to verify the reachability of $b::c::d$ and $b::d::c$ where $b$ occurs before $c$. $d::b::c$ is not included, because firing $d$ changes its state before firing $b::c$.

\subsection{Heuristic on OR Gates}\label{sectheuristic}
In the previous section, without OR gates (one state node has multiple solution successors), reachability problem is solvable in polynomial time even with AND gates if $I=O (1)$ as the reasoning of SLCG is linear with the number of local states.

In general SLCGs, there exist both AND gates and OR gates. Every disjunction may generate a fork in global reasoning, except the cases where global state satisfies all local states in a fork simultaneously and reduce the possibilities of final states and simplifies the computation. If an OR gate is followed by an AND gate, there are multiple initial states for the reasoning of the AND gate, \textit{i.e.} there are multiple sets of permutations to verify, the size of problem multiplies. 

To deal with the concern brought by OR gates, if they are numerous 
to enumerate, general SLCG can be regarded as a decision tree, where we need to make a choice at each OR gate in order to make every involved AND gate take the value 1. 

To simplify the computation, we suppose that all of the AND gates are already transformed to equivalent solution nodes. We are going to prove: under this hypothesis, the general reachability problem is at most as complex as random walks problem \cite{pearson1905problem} (see appendix \ref{appendix:A}) of the same size. 

We associate every OR gate with its predecessor and successor, more precisely, every OR gate is responsible for the reachability of its predecessor and also for the choice of its successor.

In the worst case, there is only one configuration over all disjunctions that makes objective reachable. To obtain the exact reachability by brute force search, thanks to the limit of the in-degree of OR gates, $2^D$ trials are needed with $D$ being the number of disjunctions. $2^D$ is still in exponential. To deal with such case, we transform the problem into the one with much less complexity but with a high probability of reaching the correct answer. 
\begin{theorem}\label{TheoEquiv}
The reachability problem of an SLCG with solely OR gates is equivalent to the Random Walks problem of size $\#OR$.
\end{theorem}
\begin{proof}
Numerate OR gates as $OR=\{or_1,\ldots, or_n\}$, where $or_i=\{0,1\}$. Initializing $OR$ with random configuration, every modification on an OR gate is equivalent to a ``step'' in random walks, in the worst case, after $2D^2$ trials, the probability of reaching the goal is greater than one-half, if we execute $\log_2 n$ (number of loops is adjustable) sets of trials, the probability of reaching the ``correct configuration'' is greater than $1-\frac{1}{D}$ (details in Appendix \ref{appendix:A}). If we still do not find the desired trajectory, we consider the goal is unreachable. In this case, the possibility of false negative is at most $\frac{1}{D}$ and that of false positive is $0$.
\qed\end{proof}

The proof above shows the worst case of our heuristic method. In fact, if the desired state is reachable, it is probable that exact solution is found during the trial because:
\begin{itemize}
\item Starting choice is probably not the farthest from correct choice
\item There are more than one choice that makes the desired state reachable
\end{itemize}

\subsection{Overall Process}\label{sectOverall}
Combining all the parts in section \ref{sect:4} and trajectory extraction technique in section \ref{algorithm:1}, the whole process of PermReach is shown as follows:
\begin{enumerate}
\item Precondition initial SLCG, cycles are deleted (Section \ref{sectprecond})
\item Build decision trees for AND gates and OR gates
\item Launch the heuristic on OR gates, obtain an SLCG with pure AND gates (Section \ref{sectheuristic})
\item Compute the reachability on AND gates, if reachable, quit; if not, return to step 3 (Section \ref{sectAndGates})
\end{enumerate}

\section{Implementation and Benchmarks}\label{sect:5}
The overall process in section \ref{sectOverall} is implemented on Matlab\footnote{Implementation and testing data sets are available at \url{https://github.com/XinweiChai/LCG_Reasoning}}. To evaluate the performance in large \textit{in silico} networks, we take T-cell Receptor model (TCR) \cite{saez2007logical} and epidermal growth factor receptor model (EGFR) \cite{samaga2009logic} as examples, with the former one containing 95 components and 206 transitions and the latter one containing 104 components and 389 transitions respectively. 

These models are originally Boolean networks. According to Appendix \ref{appendix:C}, they are transformable to ABANs. We then take several automata as input, varying exhaustively their initial states combinations ($2^{init\_state}$), take the reachability of the states of another automata set as output. We first test the performance of traditional model checkers, Mole\footnote{\url{http://www.lsv.fr/~schwoon/tools/mole}} and NuSMV\footnote{\url{http://nusmv.fbk.eu}}, in which Mole turns out to be timeout for 6 in 12 outputs, and all timeout for NuSMV in model EGFR. Due to the big state space, traditional model checkers are not effective. 

To validate our approach, we first use a small model: phage-$\lambda$ model \cite{thieffry1995dynamical} to compare with an alternative reachability analyzer Pint \cite{pauleve2012}. In this model with 4 components and 12 transitions (without taking consideration of the self-regulations), our result shows complete decidability while Pint is not able to figure out the reachability of $[cll=1]$. In big examples TCR and EGFR,
although PermReach takes more time than Pint, it outputs the sequence from initial state towards final state. More importantly, it gives decidable reachability for any input. In the TCR tests, PermReach gives exactly the same result as Pint did. As for EGFR tests, PermReach takes the inconclusive cases of Pint as input, and returns decidable outputs.


\begin{table}[ht]
\centering
\tiny
\begin{tabular}{|c|c|c|c|c|}
  \hline
			&	  \multicolumn{2}{c|}{TCR} & \multicolumn{2}{c|}{EGFR}  \\
\hline
Inputs	&	  \multicolumn{2}{c|}{3} & \multicolumn{2}{c|}{13}\\
\hline
  Outputs &	  \multicolumn{2}{c|}{5} & \multicolumn{2}{c|}{12} \\
\hline
  Total tests & \multicolumn{2}{c|}{$2^3\times 5=40$} & \multicolumn{2}{c|}{$2^{13}\times 12=98,304$}\\
\hline
Analyzer      &  Pint       &PermReach   &  Pint       &PermReach             \\
\hline
  True       &  \multicolumn{2}{c|}{16(40\%)}  & 64,282(65.4\%)&74,268(75.5\%)\\
\hline
  Inconclusive & \multicolumn{2}{c|}{0(0\%)}    &9,986(10.1\%)&0(0\%)  \\
\hline
  False      &  \multicolumn{2}{c|}{24(60\%)} &24,036(24.5\%)&24,036(24.5\%)\\
\hline
  Total time &  7s       &  20s        & 9h50min              & 13h20min       \\
\hline

\end{tabular}
\caption{Results of the tests on large-scale examples using Intel Core i7-3770 CPU, \@3.4GHz, 8.00G RAM. Column “Pint” gives the related results on ANs, while column “PermReach” gives the results for ABANs. “True”, “Inconclusive” and “False” give respectively the number of different results of reachability, while “Max time” and “Total time” depict respectively the maximum time of the individual computations.}
\label{tab:2}
\end{table}

As seen in the previous results, our heuristic technique is more conclusive than the reasoning of Pint. In the configuration of the heuristic approach, if there are less than 20 OR gates after preprocessing in Section \ref{sectprecond}, the computation will be shifted from heuristic to global search as the size of enumeration is acceptable. There are only 11 OR gates in EGFR model, therefore the results are firmly conclusive. Even though we do not shift to the global search, the conclusiveness is high enough according to Theorem \ref{TheoEquiv}.

To sum up, PermReach has a better time performance than traditional exhaustive model checkers (Mole and NuSMV); on the other hand, it is more conclusive than abstract analyzers (Pint) while keeping a reasonable time performance.

\section{Conclusion and future work}\label{sect:6}
This paper proposes an expressive formalism ABAN to study the reachability problem. The original approach SLCG has limited conclusiveness because static and local reasoning does not simulate all real system dynamics. Due to the complexity of global search, developing a heuristic technique based on sub-states becomes a feasible choice. The heuristic method reproduces the system dynamics by traversing possible orders of transitions. This ``dynamic tentative'' makes it closer to real dynamics than LCG is.

Future work: in the reasoning of AND gates, the computation on permutations is expensive but is still not conclusive enough, see Fig \ref{FigConflictInForks}. To speed up the whole procedure and improve the conclusiveness, we plan to apply SAT (Satisfiability) solvers or Answer Set Programming (ASP) to refine the analysis of transition orders ($\triangleright$) in the same fork and those across forks. In addition, we may contemplate the extension of our heuristic technique to multivalued models.
\appendix

\section{Random Walks problem}\label{appendix:A}
\begin{definition}[Random Walks]
Start with an arbitrary natural number $i$, at each time step, the number add or minus 1 with equal probability, how many expected steps $Z_i$ does it take to reach certain goal $n$?
\end{definition}
Starting from $n$, there is no need to move, thus the expectation of steps $\mathbb {E}[Z_n]=0$; starting from $0$, the only possibility is to move rightwards, $\mathbb {E}[Z_0]=1+\mathbb {E}[Z_1]$. Similarly, starting from $0<i<n$ we have $$\mathbb {E}[Z_i]=\frac{1}{2} (\mathbb {E}[Z_{i+1}]+1)+\frac{1}{2} (\mathbb {E}[Z_{i-1}+1])$$
With these recurrence relations, we can obtain $\mathbb {E}[Z_0]=n^2$ (worst case for $0\leq i\leq n$), \textit{i.e.} it takes $n^2$ steps on average to start from 0 to reach $n$ and less than $n$ if $i>0$. By applying Markov's inequality $\mathbb {P} (X\geq a)\leq {\frac {\mathbb {E} (X)}{a}}$, with $a=n^2$, we have $\mathbb {P}[Z_0>2n^2]\leq \frac{1}{2}$, the possibility of taking more than $n^2$ steps to reach $n$ is less than a half. The proof of Markov's inequality is shown as below:
\begin{align*}
\mathbb {E}[X]&=\sum_{k=0}^{a}k\cdot \mathbb {P}[X=k]+\sum_{k=a+1}^{\infty}k\cdot \mathbb {P}[X=k]\\
&\geq 0+2n^2\mathbb {P}[X>a]
\Rightarrow \mathbb {P}[X>a]\leq \frac {\mathbb {E} (X)}{a}
\end{align*}

For one tentative of $2n$ steps, it has at least $\frac{1}{2}$ possibility to reach $n$, if we initiate $\log_2n$ tentatives, we have at least the possibility of $1- (1-\frac{1}{2})^{\log_2n}=1-\frac{1}{n}$ to have at least one tentative reaching the goal $n$. As long as we increase the number of tentatives, the possibility of success will get closer to $1$. 

\section{Algorithm}\label{appendix:B}
The construction of an SLCG is realized by iterative updates:
\begin{algorithm}[ht]
\begin{algorithmic}
\State Initialization: 
$Ls\gets \{\omega\}$, $\mathbf{LS}\gets\{\omega\}$, $\mathbf{Sol}\gets \varnothing$
\While{$Ls\neq \varnothing$}
	\For{$a_i\in Ls$}
		\State $Ls\gets Ls\backslash a_i$
		\If{$a_i\in init\_state$}
			\State $a_i{\rm .next}=sol_{a_i}$
            \State $sol_{a_i}{\rm .next}=\varnothing$
    	\Else
    		\For{$sol=\acm{A}{a_{\lnot i}}{a_i}\in \mathbf{T}$}
    			\State $a_i{\rm .next}\gets a_i{\rm .next}\cup sol$
    			\For{$b_j\in A$}
    				\State $sol{\rm .next}\gets b_j$
    			\EndFor
    			\State $Ls\gets Ls\cup b_j$
                \State $\mathbf{LS}\gets \mathbf{LS}\cup Ls$
    		\EndFor
    		\State$\mathbf{Sol}\gets \mathbf{Sol}\cup a_i{\rm .next}$           
    	\EndIf
	\EndFor
\EndWhile
\State\Return{$(\mathbf{LS},\mathbf{Sol})$}
\end{algorithmic}
\caption{SLCG construction}
\label{algorithm:2}
\end{algorithm}
\section{Transformation from general BNs to ABANs}\label{appendix:C}

Given Boolean functions $v_i (t+1)=f_i (\mathbf{V}_i)$, with $\mathbf{V}_i$ the set of participating variables among $v_1 (t),\cdots,v_n (t)$. Boolean operators are transformable to the composition of $\lnot,\land,\lor$ (\textit{e.g.} $a\ \mathbf{XOR} \ b = (a\land \lnot b)\lor (\lnot a\land b)$), and Boolean functions possess an equivalent CNF (clausal normal form) thanks to its distributivity. As ANs interpret transitions in the way of disjunctions of conjunctions, all BNs are transformable to ANs. We can see that it does not matter whether the dynamics is synchronous or asynchronous, because these transformations are only exerted on functions/transitions.

Example: 

Let $G_B=(V,F)$ a BN with $V=\{a,b,c,d,e\}$, and has only one Boolean function, $F=\{f (a)= (b\lor c)\land (d\lor e)\}$, we have 
$f (a)= (b\land d)\lor (b\land e)\lor (c\land d)\lor (c\land e)$, and $\lnot f (a)= (\lnot b\land \lnot c)\lor (\lnot d\land \lnot e)$. The equivalent AN is then constructed: 5 automata $\mathbf{\Sigma}=\{a,b,c,d,e\}$, with transitions: $\mathbf{T}=\{\acm{\{b_1,d_1\}}{a_0}{a_1},\ \acm{\{b_1,e_1\}}{a_0}{a_1},\ \acm{\{c_1,d_1\}}{a_0}{a_1},\ \acm{\{c_1,e_1\}}{a_0}{a_1},\ \acm{\{b_0,c_0\}}{a_1}{a_0},\ \acm{\{d_0,e_0\}}{a_1}{a_0}\}$.

%
%


\bibliographystyle{spmpsci}      
\bibliography{bib}   

%
%

\end{document}